\documentclass[a4paper,UKenglish,cleveref, autoref, thm-restate]{lipics-v2021}

\pdfoutput=1 %
\hideLIPIcs  %

\newcommand{\PSAC}{\textsc{PSAC}}
\newcommand{\PSACdefault}{\textsc{PSAC-default}}
\newcommand{\PSACfast}{\textsc{PSAC-fast}}
\newcommand{\DCX}{\textsc{DCX}}
\newcommand{\SA}{\textsc{SA}}

\newcommand{\CC}{\textsc{CC}}
\newcommand{\Wiki}{\textsc{Wiki}}
\newcommand{\DNA}{\textsc{DNA}}
\newcommand{\divop}{\;\text{div}\;}
\newcommand{\modop}{\;\text{mod}\;}

\usepackage{dsfont}
\newcommand{\doublestroke}[1]{\mathbb{#1}}

\usepackage{physics}
\usepackage{booktabs}
\usepackage{cleveref}
\usepackage[linesnumbered,ruled,vlined]{algorithm2e}

\usepackage[inline]{enumitem}

\usepackage{tikz}
\usepackage{tikz-cd}
\usepackage{pgfplots}
\usepackage{graphicx}
\pgfplotsset{compat=1.14}

\tikzset{
  double color fill/.code 2 args={
    \pgfdeclareverticalshading[%
    tikz@axis@top,tikz@axis@middle,tikz@axis@bottom%
    ]{diagonalfill}{100bp}{%
      color(0bp)=(tikz@axis@bottom);
      color(50bp)=(tikz@axis@bottom);
      color(50bp)=(tikz@axis@middle);
      color(50bp)=(tikz@axis@top);
      color(100bp)=(tikz@axis@top)
    }
    \tikzset{shade, left color=#1, right color=#2, shading=diagonalfill}
  }
}

\usepackage{xcolor}
\definecolor{my-dark-red}{RGB}{183, 28, 28}
\definecolor{my-red}{RGB}{244,67,54}
\definecolor{my-pink}{RGB}{233,30,99}
\definecolor{my-purple}{RGB}{156,39,176}
\definecolor{my-deep-purple}{RGB}{103,58,183}
\definecolor{my-indigo}{RGB}{63,81,181}
\definecolor{my-blue}{RGB}{33,150,243}
\definecolor{my-light-blue}{RGB}{3,169,244}
\definecolor{my-cyan}{RGB}{0,188,212}
\definecolor{my-teal}{RGB}{0,150,136}
\definecolor{my-green}{RGB}{76,175,80}
\definecolor{my-light-green}{RGB}{139,195,74}
\definecolor{my-lime}{RGB}{205,220,57}
\definecolor{my-yellow}{RGB}{255,235,59}
\definecolor{my-amber}{RGB}{255,193,7}
\definecolor{my-orange}{RGB}{255,152,0}
\definecolor{my-deep-orange}{RGB}{255,87,34}
\definecolor{my-brown}{RGB}{121,85,72}
\definecolor{my-grey}{RGB}{158,158,158}
\definecolor{my-blue-grey}{RGB}{96,125,139}
\definecolor{my-lipics-grey}{rgb}{0.6,0.6,0.61}

\title{Fast and Lightweight Distributed Suffix Array Construction -- First Results}

\author{Manuel Haag}{Karlsruhe Institute of Technology, Germany}{uozeb@student.kit.edu}{}{}
\author{Florian Kurpicz}{Karlsruhe Institute of Technology, Germany}{kurpicz@kit.edu}{https://orcid.org/0000-0002-2379-9455}{}
\author{Peter Sanders}{Karlsruhe Institute of Technology, Germany}{sanders@kit.edu}{https://orcid.org/0000-0003-3330-9349}{}
\author{Matthias Schimek}{Karlsruhe Institute of Technology, Germany}{schimek@kit.edu}{https://orcid.org/0009-0002-6402-9016}{}

\authorrunning{M. Haag, F. Kurpicz, P. Sanders, and M. Schimek} %

\Copyright{Jane Open Access and Joan R. Public} %



\category{} %

\relatedversion{} %

\nolinenumbers %

\EventEditors{John Q. Open and Joan R. Access}
\EventNoEds{2}
\EventLongTitle{42nd Conference on Very Important Topics (CVIT 2016)}
\EventShortTitle{CVIT 2016}
\EventAcronym{CVIT}
\EventYear{2016}
\EventDate{December 24--27, 2016}
\EventLocation{Little Whinging, United Kingdom}
\EventLogo{}
\SeriesVolume{42}
\ArticleNo{23}

\begin{document}

\maketitle

\begin{abstract}
    We present first algorithmic ideas for a practical and lightweight adaption of the DCX suffix array construction algorithm [Sanders~et~al.,~2003] to the distributed-memory setting.
    Our approach relies on a bucketing technique which enables a lightweight implementation which uses less than half of the memory required by the currently fastest distributed-memory suffix array algorithm PSAC [Flick~and~Aluru,~2015] while being competitive or even faster in terms of running time.
\end{abstract}

\section{Introduction}
The suffix array~\cite{DBLP:conf/soda/ManberM90,GonnetBS1992PatArray} is one of the most well-studied text indices.
Given a text \(T\) of length \(n\), the suffix array \SA\ simply lists the order of the lexicographically sorted suffixes, i.e., for all \(1\leq i\leq j\leq n\) we have \(T[\SA[i]..n]\leq T[\SA[j]..n]\).
To compute the suffix array, we have to (implicitly) sort all suffixes of the text.
Therefore, the task of constructing the suffix array is sometimes referred to as suffix sorting.
Even though, we have to consider all suffixes of the text, whose total length is quadratic in the size of the input, suffix arrays can be constructed in linear time requiring only constant working space in addition to the space for the suffix array~\cite{Goto2019OptSACA,LiLH18OptSACA}.

Despite their simplicity, suffix arrays have numerous applications in pattern matching and text compression~\cite{Ohlebusch2013Bioinformatics}.
They are a very powerful full-text index and are used as a space-efficient replacement~\cite{AbouelhodaKO2004EnhancedSA} of the suffix tree---one of the most powerful full-text indices.
Furthermore, suffix arrays can be used to compute the Burrows-Wheeler transform~\cite{BurrowsW1994BWT}, which is the backbone of many compressed full-text indices~\cite{FerraginaM2000FMIndex,GagieNP2020FullyFunctionalRIndex}.

In today's information age, the amount of textual data that has to be processed is ever-increasing with no sign of slowing down.
For example, the English Wikipedia contains around 60 million pages and grows by around 2.5 million pages each year.\footnote{\url{https://en.wikipedia.org/wiki/Wikipedia:Size_of_Wikipedia}, last accessed 2024-12-11.}
A snapshot of all public source code repositories created by over 100 million developers\footnote{\url{https://github.blog/news-insights/company-news/100-million-developers-and-counting/}, last accessed 2024-12-11.} on GitHub requires more than 21\,TB to store\footnote{\url{https://archiveprogram.github.com/arctic-vault/}, last accessed 2024-12-11.}.
Furthermore, the capability to sequence genomic data is increasing exponentially, due to technical advances~\cite{Stephens2015BigDataGenomical}.
All these examples show the importance of scaling algorithms for the analysis of textual information many of which use the suffix array as building block.

One possible solution to tackle this problem is to use distributed algorithms.
In the distributed-memory setting, we can utilize many processing elements (PEs) that are connected via a network, e.g., high-performance clusters or cloud computing.
In this setting, the main obstacle when computing suffix arrays is the immense amount of working memory required by the current state-of-the-art algorithms.
Even carefully engineered implementations require around \(30\times\)--\(60\times\) the input size as working space~\cite{DBLP:conf/alenex/0001K19,DBLP:conf/sc/FlickA15}.
Additionally, there is a significant space-time trade-off.
The memory-efficient algorithms tend to be slower.
We thus ask the question:

\begin{center}
\em
Is there a scaling, fast, \emph{and} memory-efficient suffix array construction algorithm in distributed memory?
\end{center}

\subparagraph*{Structure of this Paper.}
First, in \cref{sec:preliminaries}, we introduce some basic concepts required for suffix array construction and distributed-memory algorithms.
\Cref{sec:related-work} discusses previous work on suffix array construction.

In \cref{sec:distributed-dcx}, we start with a description of the distributed-memory variant of the DCX~\cite{DBLP:journals/jacm/KarkkainenSB06} suffix array construction algorithm.
In \cref{sec:bucketing}, we demonstrate how our previously developed technique for space-efficient string sorting \cite{PascalStringSorting, DBLP:conf/esa/KurpiczM0S24} can be applied to the DCX suffix sorting to obtain a more lightweight algorithm.
Subsequently, we introduce a \emph{randomized chunking scheme} to provide provable load-balancing guarantees of our space-efficient (suffix) sorting approach.
As a side-result of independent interest, in \cref{sec:external-memory}, we briefly describe how the algorithm can be extended to the distributed external-memory model.
Finally, preliminary results of a first prototypical implementation of our ideas using MPI are presented in \Cref{sec:evaluation} followed by a brief outline of our future work in \Cref{sec:conclusion}.

\subparagraph*{Summary of our Contributions.}
The main contributions that we are presenting in this paper are
\begin{enumerate*}[label=(\roman*)]
\item a scaling, fast, \emph{and} space-efficient distributed-memory suffix array construction algorithm, using
\item a new randomized chunking scheme for load balancing, that
\item can also be applied to other (distributed) models of computation and algorithms.
\end{enumerate*}

\section{Preliminaries}
\label{sec:preliminaries}

We assume a distributed-memory machine model consisting of $p$ processing elements (PEs) allowing single-ported point-to-point communication.
The cost of exchanging a message of $h$ \emph{machine words} between any two PEs is $\alpha + \beta h$, where $\alpha$ accounts for the message start-up overhead and $\beta$ quantifies the time to exchange one machine word.
\begin{table}[t]
	\centering
	\caption[List of used symbols]{Symbols used in this paper.}
	\label{tab:symbols_and_abbreviations}
	\begin{tabular}{ll}
		\multicolumn{2}{c}{\textbf{Symbols}}            \\
		\toprule
		$ p$       & number of processing elements (PEs)      \\
		$ n$       & total length of the input text      \\
		$ \sigma$  & size of the alphabet             \\
                $ \$$      & sentinel character with $\$ < c$ for $c \in \Sigma$ \\
                $ X$       & size of the difference cover in the DCX algorithm \cite{DBLP:journals/jacm/KarkkainenSB06}            \\
		$ q$       & number of buckets      \\
		$ c$       & size of a chunk      \\
		\bottomrule
	\end{tabular}
\end{table}
The input to our algorithms is a text $T$ consisting of $n-1$ characters over an alphabet $\Sigma$.
By $T[i]$, we denote the $i$-th character of $T$ for $0 \le i < n -1$.
We assume $T[n]$ to be a sentinel character $\$ \notin \Sigma$ with $\$ < z$ for all $z \in \Sigma$.
The $i$-th suffix of $T$, $s_i = T[i, n]$, is the substring starting at the $i$-th character of $T$.
Due to the sentinel element all suffixes are prefix free.
The \emph{suffix array} SA contains the lexicographical ordering of all suffixes of $T$. More precisely,
\SA{} is an array of length $n$ with $\SA[i]$ containing the index of the $i$-th smallest suffix of $T$.
A length-$l$- (or simply $l$)-prefix of a suffix with starting position $i$ is the substring $T[i,l)$.

In our distributed setting, we assume that each PE $i$ obtains a local subarray $T_i$ of $T$ as input such that $T$ is the concatenation of all local input arrays $T_i$.
Furthermore, we assume the input to be well-balanced, i.e., $\abs{T_i} = \Theta(n/p)$.
For our DCX algorithm, we assume a suitable padding of up to $X$ sentinel characters at the end of the text.

\section{Related Work}
\label{sec:related-work}
There has been extensive research on the construction of suffix arrays  in the sequential, external-memory, shared-memory parallel, and (to a somewhat lesser extent) also in the distributed-memory setting.
All suffix array construction algorithms are based on three general algorithmic techniques:
\begin{enumerate*}[label=(\arabic*)]
\item prefix doubling,
\item induced copying, and
\item recursion.
\end{enumerate*}

In the following, we give a brief overview of these techniques.
Since---to the best of our knowledge---all external, shared, and distributed-memory algorithms have a sequential counterpart, we focus on the sequential and distributed-memory algorithms.
See \cref{fig:timeline_suffix_array_construction} for the evolution of sequential suffix array construction algorithms.
For a more comprehensive overview, we refer to the most recent surveys on suffix array construction~\cite{DBLP:phd/dnb/Bingmann18,DBLP:books/sp/22/BingmannD0KO022}.

\begin{figure}[t]
  \centering
       \begin{tikzpicture}[transform shape]
        \draw[ultra thick,rounded corners=5pt,draw=my-blue-grey,fill=my-grey!25]
        (0,-3.05) rectangle ++(2.25,12.2);
        \draw[ultra thick,rounded corners=5pt,draw=my-blue-grey,fill=my-grey!25]
        (2.2775,-3.05) rectangle ++(4.25,12.2);
        \draw[ultra thick,rounded corners=5pt,draw=my-blue-grey,fill=my-grey!25]
        (6.555,-3.05) rectangle ++(4.25,12.2);
          
         (.25,-.375) rectangle ++(4.2,1);

        \path[rounded corners=5pt, fill=my-green, ] (-1, 8.0) rectangle ++(2.9, .95);
        \node at (-.5, 8.475) {1990};
        \path[rounded corners=5pt, fill=my-green!75, ] (-1, 7.0) rectangle ++(6.5, .95);
        \node at (-.5, 7.475) {1999};

        \path[rounded corners=5pt, fill=my-green!50, ] (-1, 6.4) rectangle ++(3.9, .55);
        \path[rounded corners=5pt,fill=my-grey!25] (2.4, 6.7) rectangle ++(1.45, .95);
        \path[rounded corners=5pt,fill=my-grey!25] (2.35, 6.95) rectangle ++(1.55, .75);
        \path[fill=my-grey!25] (2.35, 6.95) rectangle ++(1.55, .25);
        \path[rounded corners=5pt,fill=my-green!50, ] (2.4, 6.7) rectangle ++(1.45, .95);
        \node at (-.5, 6.675) {2000};

        \path[rounded corners=5pt, fill=my-blue, ] (-1, 5.65) rectangle ++(5, .7);
        \path[rounded corners=5pt,fill=my-grey!25] (2.95, 5.65) rectangle ++(1.95, .95);
        \path[rounded corners=5pt, fill=my-blue, ] (2.95, 5.65) rectangle ++(1.95, 1);
        \node at (-.5, 6) {2002};

        \path[rounded corners=5pt, fill=my-blue!75, ] (-1, 5.05) rectangle ++(1.5, .55);
        \path[fill=my-grey!25] (.25, 5.05) rectangle ++(.25, .55);
        \path[fill=my-blue!75, ] (.25, 5.05) rectangle ++(5.8, .55);
        \path[rounded corners=5pt,fill=my-grey!25] (5.65, 5.35) rectangle ++(.5, 1.7);
        \path[rounded corners=5pt,fill=my-blue!75, ] (5.65, 5.35) rectangle ++(4.95, 1.75);
        \node at (-.5, 5.325) {2003};

        \path[rounded corners=5pt, fill=my-blue!50, ] (-1, 4.45) rectangle ++(1.5, .55);
        \path[rounded corners=5pt,fill=my-grey!25] (4.85, 4.4) rectangle ++(1.4, .95);
        \path[rounded corners=5pt,fill=my-blue!50, ] (4.9, 4.45) rectangle ++(5.7, .85);
        \path[fill=my-grey!25] (.25, 4.75) rectangle ++(.25, .25);
        \path[fill=my-grey!25] (4.25, 4.75) rectangle ++(1.25, .25);
        \path[fill=my-blue!50, ] (.25, 4.75) rectangle ++(5.7, .25);
        \node at (-.5, 4.725) {2004};

        \path[rounded corners=5pt, fill=my-deep-orange, ] (-1, 3.85) rectangle ++(1.55, .55);
        \path[rounded corners=5pt, fill=my-grey!25] (1.1, 4.15) rectangle ++(.25, .5);
        \path[rounded corners=5pt, fill=my-deep-orange, ] (.55, 4.15) rectangle ++(4.3, .55);
        \path[fill=my-grey!25] (.25, 4.15) rectangle ++(.5, .25);
        \path[fill=my-deep-orange, ] (.25, 4.15) rectangle ++(.5, .25);
        \path[rounded corners=5pt,fill=my-grey!25] (7.35, 3.9) rectangle ++(1.55, 1);
        \path[fill=my-grey!25](4.6,4.15) rectangle ++(.5,.25);
        \path[fill=my-grey!25](4.1,4.15) rectangle ++(.5,.25);
        \path[fill=my-deep-orange, ] (4.1, 3.9) rectangle ++(3.5, .5);
        \path[rounded corners=5pt,fill=my-grey!25] (7.4, 3.9) rectangle ++(1.45, .95);
        \path[rounded corners=5pt,fill=my-deep-orange, ] (7.4, 3.9) rectangle ++(1.45, .95);
        \node at (-.5, 4.125) {2005};

        \path[rounded corners=5pt, fill=my-deep-orange!75, ] (-1, 3.25) rectangle ++(1.5, .55);
        \path[rounded corners=5pt, fill=my-deep-orange!75, ] (.6, 3.6) rectangle ++(2.25, .5);
        \path[rounded corners=5pt,fill=my-grey!25] (2.35, 3.55) rectangle ++(1.8, 1.05);
        \path[rounded corners=5pt,fill=my-deep-orange!75, ] (2.4, 3.7) rectangle ++(1.7, .85);
        \path[fill=my-grey!25] (.25, 3.55) rectangle ++(2.5, .25);
        \path[fill=my-deep-orange!75, ] (.25, 3.55) rectangle ++(2.5, .25);
        \path[fill=my-grey!25] (2.35, 3.6) rectangle ++(.5, .5);
        \path[fill=my-deep-orange!75, ] (2.35, 3.6) rectangle ++(.5, .5);
        \path[rounded corners=5pt,fill=my-grey!25] (2.55, 3.6) rectangle ++(1.55, .5);
        \path[rounded corners=5pt,fill=my-deep-orange!75, ] (2.55, 3.6) rectangle ++(1.55, .5);
        \node at (-.5, 3.525) {2006};

        \path[fill=my-grey!25] (.5, 3.35) rectangle ++(1.95, .25);
        \path[rounded corners=5pt,fill=my-grey!25] (.55, 2.7) rectangle ++(3.4, .9);
        \path[rounded corners=5pt,fill=my-deep-orange!50, ] (.55, 2.7) rectangle ++(7.65, .85);
        \path[rounded corners=5pt, fill=my-deep-orange!50, ] (-1, 2.7) rectangle ++(1.95, .5);
        \node at (-.5, 2.925) {2007};

        \path[rounded corners=5pt, fill=my-light-green, ] (-1, 2.05) rectangle ++(11.6, .6);
        \path[rounded corners=5pt,fill=my-grey!25] (8.65, 2.25) rectangle ++(1.95, 1.15);
        \path[rounded corners=5pt,fill=my-light-green, ] (8.65, 2.25) rectangle ++(1.95, 1.1);
        \node at (-.5, 2.325) {2008};

        \path[rounded corners=5pt, fill=my-light-green!75, ] (-1, 1.45) rectangle ++(8.6, .55);
        \path[rounded corners=5pt,fill=my-grey!25] (5.45, 1.75) rectangle ++(2.35, .825);
        \path[rounded corners=5pt,fill=my-light-green!75, ] (5.5, 1.75) rectangle ++(2.25, .8);
        \path[rounded corners=5pt, fill=my-grey!25] (7.1, 1.45) rectangle ++(.5, .55);
        \path[rounded corners=5pt, fill=my-light-green!75, ] (7.1, 1.45) rectangle ++(.65, .65);
        \path[rounded corners=5pt, fill=my-grey!25] (5.4, 1.45) rectangle ++(.5, .55);
        \path[rounded corners=5pt, fill=my-light-green!75, ] (5.4, 1.45) rectangle ++(.5, .55);
        \node at (-.5, 1.725) {2009};

        \path[rounded corners=5pt, fill=my-light-green!50, ] (-1, .65) rectangle ++(8.55, .75);
        \path[rounded corners=5pt,fill=my-grey!25] (5.65, .7) rectangle ++(1.95, .875);
        \path[rounded corners=5pt,fill=my-light-green!50, ] (5.7, .7) rectangle ++(1.85, .85);
        \path[rounded corners=5pt, fill=my-grey!25] (5.45, .65) rectangle ++(.55, .75);
        \path[rounded corners=5pt, fill=my-light-green!50, ] (5.45, .65) rectangle ++(.55, .75);
        \node at (-.5, 1.025) {2011};

        \path[rounded corners=5pt,fill=my-cyan, ] (-1, -.35) rectangle ++(8.35, .95);
        \node at (-.5, .125) {2016};

        \path[rounded corners=5pt,fill=my-cyan!75, ] (-1, -1.35) rectangle ++(8.85, .95);
        \node at (-.5, -.9) {2017};

        \path[rounded corners=5pt,fill=my-cyan!50, ] (-1, -2.35) rectangle ++(8.35, .95);
        \node at (-.5, -1.9) {2021};
        
        \draw[ultra thick,rounded corners=5pt,draw=my-blue-grey]
        (0,-3.05) rectangle ++(2.25,12.2);
        \draw[ultra thick,rounded corners=5pt,draw=my-blue-grey]
        (2.2775,-3.05) rectangle ++(4.25,12.2);
        \draw[ultra thick,rounded corners=5pt,draw=my-blue-grey]
        (6.555,-3.05) rectangle ++(4.25,12.2);

        \node at (1.125,-2.7) {%
          \normalsize\begin{tabular}{l}\textbf{Prefix}\\[-.145cm]\textbf{Doubling}\end{tabular}};
        \node at (4,-2.7) {%
          \normalsize\begin{tabular}{l}~\\[-.145cm]\textbf{Induced Copying}\end{tabular}};
        \node at (8.75,-2.7) {%
          \normalsize\begin{tabular}{l}~\\[-.145cm]\textbf{Recursion}\end{tabular}};

        \draw[rounded corners=5pt,fill=my-grey!15]
        (0.5, 8.1) rectangle ++(1.25, .75) node[pos=.5] (MM) {%
          \scriptsize \begin{tabular}{c}\cite{DBLP:conf/soda/ManberM90}\\original\end{tabular}};
        \draw[rounded corners=5pt,double color fill={my-grey!15}{my-brown!75},shading angle=45]
        (0.5, 7.1) rectangle ++(1.25, .75) node[pos=.5] (LS) {%
          \scriptsize \begin{tabular}{c}\cite{LarssonS2007FasterSuffixSorting}\\qsufsort\end{tabular}};

        \path[-latex, draw] (1.125,8.1) -- (1.125,7.85);
        \path[draw] (1.125,7.1) edge[out=270,in=90,-latex] (2.2675,3.5);

        \draw[rounded corners=5pt,double color fill={my-grey!15}{my-brown!75},shading angle=45]
        (1.75, 2.75) rectangle ++(1, .75) node[pos=.5] {%
          \scriptsize \begin{tabular}{c}\cite{SchurmannS2007IncomplexSACA}\\bpr\end{tabular}};

        \draw[dashed, rounded corners=5pt,fill=my-grey!15]
        (3.75, 8.1) rectangle ++(1.25, .75) node[pos=.5] {%
          \scriptsize \begin{tabular}{c}\cite{BurrowsW1994BWT}\\BWT\end{tabular}};
        \draw[rounded corners=5pt,fill=my-grey!15]
        (2.5, 6.8) rectangle ++(1.25, .75) node[pos=.5] {%
          \scriptsize \begin{tabular}{c}\cite{Seward2000BWTSACA}\\1/2 copy\end{tabular}};
        \draw[rounded corners=5pt,fill=my-grey!15]
        (4, 7.1) rectangle ++(1.35, .75) node[pos=.5] {%
          \scriptsize \begin{tabular}{c}\cite{ItohT1999ABCopySACA}\\A/B copy\end{tabular}};
        \draw[rounded corners=5pt,double color fill={my-grey!15}{my-brown!75},shading angle=45]
        (3.05, 5.75) rectangle ++(1.75, .75) node[pos=.5] {%
          \scriptsize \begin{tabular}{c}\cite{ManziniF2004Lightweight}\\deep-shallow\end{tabular}};
        \draw[rounded corners=5pt,double color fill={my-grey!15}{my-brown!75},shading angle=45]
        (5, 4.5) rectangle ++(1, .75) node[pos=.5] {%
          \scriptsize \begin{tabular}{c}\cite{Manzini2004}\\chains\end{tabular}};
        \draw[rounded corners=5pt,double color fill={my-grey!15}{my-brown!75},shading angle=45]
        (2.5, 3.75) rectangle ++(1.5, .75) node[pos=.5] {%
          \scriptsize \begin{tabular}{c}\cite{Mori2015libdivsufsort}\\DivSufSort\end{tabular}};
        \draw[rounded corners=5pt,fill=my-grey!15]
        (4.5, 2.75) rectangle ++(1.75, .75) node[pos=.5] {%
          \scriptsize \begin{tabular}{c}\cite{ManiscalcoP2007CacheAwareSACA}\\cache aware\end{tabular}};

        \path[draw] (3.75, 8.5) -- (3.375,8.5) edge[out=180,in=90,-] (3.125,8.25);
        \path[-latex,draw] (3.125,8.25) -- (3.125,7.55);

        \path[-latex,draw] (4.675, 8.1) -- (4.675,7.85);

        \path[draw] (5, 8.5) -- (6.155, 8.5) edge[out=0,in=90,-] (6.375,8.25);
        \path[-latex,draw] (6.375,8.25) -- (6.375,7);

        \path[draw] (3,6.8) edge[out=270,in=90,-latex] (2.2675,3.5);
        \path[-latex,draw] (3.4,6.8) -- (3.4,6.5);

        \path[draw] (4.9, 7.1) -- (4.9, 5.75) edge[out=270, in=90,-latex] (3.75, 4.5);
        \path[draw,-latex] (3.4,5.75) -- (3.4, 4.5);

        \path[draw] (5, 8.25) -- (5.25,8.25) edge[out=0, in=90,-] (5.5,8);
        \path[-latex,draw] (5.5,8) -- (5.5,5.25);

        \path[-latex,draw] (5.5,4.5) -- (5.5,3.5);

        \path[draw] (5.125, 7.1) -- (5.125, 6) edge[out=270, in=180,-] (5.375, 5.775);
        \path[-latex,draw] (5.375, 5.775) -- (6.25,5.775);

        \path[draw] (3.4,3.75) -- (3.4,2.325) edge[out=270, in=180,-] (3.65, 2.075);
        \path[-latex,draw] (3.65, 2.075) -- (5.55,2.075);

        \draw[rounded corners=5pt,fill=my-grey!15]
        (5.75, 6.25) rectangle ++(1.25, .75) node[pos=.5] {%
          \scriptsize \begin{tabular}{c}\cite{BurkhardtK2003DiffCover}\\diffcover\end{tabular}};
        \draw[rounded corners=5pt,double color fill={my-blue-grey!75}{my-brown!75},shading angle=45]
        (5.55, 1.75) rectangle ++(2.15, .75) node[pos=.5] {%
          \scriptsize \begin{tabular}{c}\cite{Mori2008SAIS,NongZH2011SADS}\\SAIS/SADS\end{tabular}};
        \draw[rounded corners=5pt,double color fill={my-blue-grey!75}{my-brown!75},shading angle=45]
        (5.75, .75) rectangle ++(1.75, .75) node[pos=.5] {%
          \scriptsize \begin{tabular}{c}\cite{Nong2013ConstSpaceSACA}\\SACA-K\end{tabular}};

        \draw[rounded corners=5pt,fill=my-blue-grey!75]
        (5.5, -.25) rectangle ++(1.75, .75) node[pos=.5] {%
          \scriptsize \begin{tabular}{c}\cite{LiLH2022ConstSpaceSACA}\\$O(1)$ space\end{tabular}};

        \draw[rounded corners=5pt,fill=my-blue-grey!75]
        (6, -1.25) rectangle ++(1.75, .75) node[pos=.5] {%
          \scriptsize \begin{tabular}{c}\cite{Goto2019ConstSpaceSACA}\\$O(1)$ space\end{tabular}};

        \path[-latex,draw] (7,6.8) -- (7.5,6.8);
        \path[-latex,draw] (7,.75) -- (7,.5);
        \path[-latex,draw] (7.325,.75) -- (7.325,-.5);

        \draw[dashed, rounded corners=5pt,fill=my-grey!15]
        (8, 8.1) rectangle ++(1.45, .75) node[pos=.5] {%
          \scriptsize \begin{tabular}{c}\cite{Farach1997OptSuffixTree}\\$O(n)$ tree\end{tabular}};

        \draw[rounded corners=5pt,double color fill={my-blue-grey!75}{my-brown!75}, shading angle=45]
        (7.5, 6.25) rectangle ++(1.25, .75) node[pos=.5] {%
          \scriptsize \begin{tabular}{c}\cite{DBLP:journals/jacm/KarkkainenSB06}\\DC3\end{tabular}};
        \draw[rounded corners=5pt,fill=my-blue-grey!75]
        (9, 6.25) rectangle ++(1.5, .75) node[pos=.5] {%
          \scriptsize \begin{tabular}{c}\cite{KimSPP2005LinearTimeSACA}\\mod2 split\end{tabular}};

        \draw[rounded corners=5pt,fill=my-blue-grey!75]
        (8.35, 5.4) rectangle ++(1.25, .75) node[pos=.5] {%
          \scriptsize \begin{tabular}{c}\cite{Hon2009Mod2SACA}\\mod2\end{tabular}};

        \draw[rounded corners=5pt,fill=my-blue-grey!75]
        (6.25, 5.4) rectangle ++(1.25, .75) node[pos=.5] {%
          \scriptsize \begin{tabular}{c}\cite{KoA2005SpaceEfficientSACA}\\L/S split\end{tabular}};

        \draw[rounded corners=5pt,fill=my-blue-grey!75]
        (7.5, 4) rectangle ++(1.25, .75) node[pos=.5] {%
          \scriptsize \begin{tabular}{c}\cite{Na2005}\\succinct\end{tabular}};
        \draw[rounded corners=5pt,fill=my-blue-grey!75]
        (9.125, 4.5) rectangle ++(1.25, .75) node[pos=.5] {%
          \scriptsize \begin{tabular}{c}\cite{KimJP2004FixedSizeSACA}\\fixed $\Sigma$\end{tabular}};

        \draw[rounded corners=5pt,fill=my-blue-grey!75]
        (6.35, 2.75) rectangle ++(1.75, .75) node[pos=.5] {%
          \scriptsize \begin{tabular}{c}\cite{NongZ2007OptimalConstantAlphabetSACA}\\$O(n\lg\vert\Sigma\vert)$\end{tabular}};
        \draw[rounded corners=5pt,fill=my-blue-grey!75]
        (8.75, 2.5) rectangle ++(1.75, .75) node[pos=.5] {%
          \scriptsize \begin{tabular}{c}\cite{AdjerohN2010SFESACA}\\SFE-coding\end{tabular}};

        \path[draw] (7,5.4) -- (7,4.5) edge[out=270, in=90,-latex] (6, 3.5);
        \path[-latex,draw] (7.225, 5.4) -- (7.225,3.5);

        \path[-latex,draw] (7,2.75) -- (7,2.5);
        \path[-latex,draw] (7,1.75) -- (7,1.5);

        \path[draw] (8,8.5) -- (7.475,8.5) edge[out=180, in=90] (7.225,8.25);
        \path[-latex,draw] (7.225,8.25) -- (7.225,6.15);

        \path[draw] (8.5,8.1) edge[out=270, in=90,-latex] (8.125,7);
        \path[draw,-latex] (8.875,8.1) -- (8.875,6.15);
        \path[draw] (8.95,8.1) edge[out=270, in=90,-latex] (9.75,7);

        \path[-latex,draw] (8.125,6.25) -- (8.125,4.75);
        \path[-latex,draw] (9.75,6.25) -- (9.75,5.25);

        \path[draw] (8.25,6.25) -- (8.25,5.25) edge[out=270, in=180] (8.5,5);
        \path[draw] (8.5,5) -- (8.65,5) edge[out=0, in=90] (8.9,4.75);

        \path[draw] (8.9,4.5) -- (8.9, 3.6) edge[out=270, in=0] (8.65,3.35);
        \path[-latex,draw] (8.65,3.35) -- (8.1,3.35);

        \path[draw] (8.9,4.75) -- (8.9, 4.25) edge[out=270, in=180] (9.15,4);
        \path[draw] (9.15,4) -- (9.375,4) edge[out=0, in=90] (9.625,3.75);
        \path[-latex,draw] (9.625,3.75) -- (9.625,3.25);

        \draw[rounded corners=5pt,double color fill={my-blue-grey!75}{my-brown!75},shading angle=45]
        (2.4, -.25) rectangle ++(1.75, .75) node[pos=.5] {%
          \scriptsize \begin{tabular}{c}\cite{Baier2016GSACA}\\GSACA\end{tabular}};

        \path[draw] (4,4.15) edge[out=0, in=90] (4.25,3.9);
        \path[draw] (4.25, 3.9) -- (4.25, -1.5);
        \path[draw] (4.25, -1.5) edge[out=270, in=180] (4.5, -1.75);
        \path[-latex,draw] (4.5,-1.75) -- (5.5,-1.75);

        \path[draw] (4.5, 1.8) -- (5.65, 1.8);
        \path[draw] (4.5, 1.8) edge[out=180, in=90] (4.25,1.55);

        \path[draw] (4.5, 4.75) -- (5, 4.75);
        \path[draw] (4.5, 4.75) edge[out=180, in=90] (4.25,4.5);
        \path[draw] (4.25, 4.5) -- (4.25, 3.9);

        \draw[rounded corners=5pt,double color fill={my-blue-grey!75}{my-brown!75},shading angle=45]
        (5.5, -2.25) rectangle ++(1.75, .75) node[pos=.5] {%
          \scriptsize \begin{tabular}{c}\cite{Grebnov2021libsais}\\libSAIS\end{tabular}};
        
      \end{tikzpicture}
      \caption{Timeline of \emph{sequential} suffix array construction with algorithms that share techniques are marked with an arrow. Figure based on~\cite{DBLP:journals/csur/PuglisiST07,DBLP:phd/dnb/Bingmann18,DBLP:phd/dnb/Kurpicz20}. The three techniques are shown as columns and algorithms that combine multiple techniques are crossing the borders. Suffix array construction algorithms with linear running time are highlighted in dark gray. If an implementation is publicly available, the algorithm is also marked in brown.}
      \label{fig:timeline_suffix_array_construction}
\end{figure}
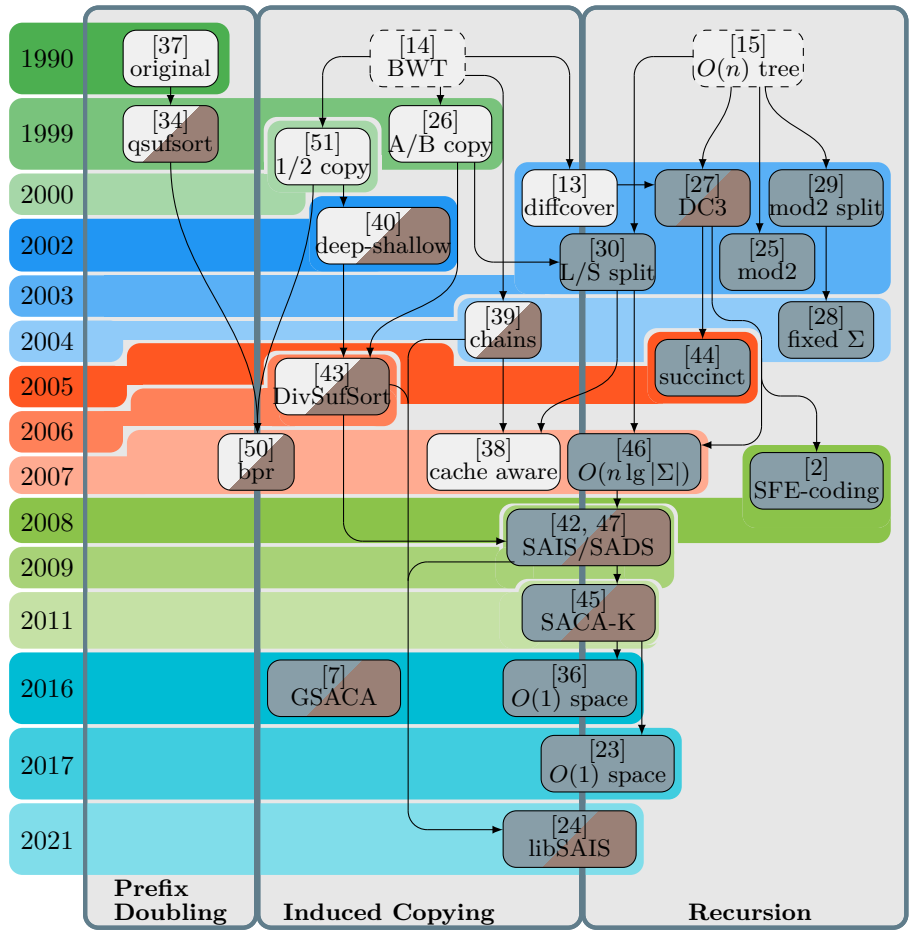

\subparagraph*{Prefix-Doubling.}
In algorithms based on prefix doubling, the suffixes are iteratively sorted by their lenght-$h$ prefix starting with $h=1$.
Now, all suffixes that share a common $h$-prefix are said to be in the same $h$-group and have an $h$-rank corresponding to the number of suffixes in lexicographically smaller $h$-groups.
By sorting all suffixes based on their $h$-group, we can compute the corresponding suffix array $\SA_h$.
Note that this suffix array does not necessarily have to be unique, as the order of suffixes within an $h$-group is not unique.
If for some $h$, all $h$-groups contain only a single suffix, i.e., the $h$-ranks of all suffixes are unique, then we have $\SA_h=\SA$.
Therefore, the idea is to increase $h$ until, all $h$-ranks are unique.
To this end, during each iteration, the length of the considered prefixes is doubled.
Fortunately, we do not have to compare the prefixes explicitly.
Instead, during iteration $i > 0$, for a suffix starting at index $j$ the rank by its length-$h$ prefix can be inferred by sorting the ranks $(\mathrm{rank}_{h/2}[j], \mathrm{rank}_{h/2}[j+h/2])$ of the suffixes $j$ and $j + h/2$ computed in the previous iteration.
Using the overlap of suffixes in a text, prefix-doubling boils down to at most $\mathcal{O}(\log{n})$ rounds in which $n$ pairs of integers have to be sorted.
Thus, this approach has an overall complexity in $\mathcal{O}(n \log{n})$ in the sequential setting, when using integer sorting algorithms.
A prefix doubling algorithm is the original suffix array construction algorithm~\cite{DBLP:conf/soda/ManberM90}.
However, in the sequential setting, this approach has not received much attention, due to its inherent non-linear running time.
However, in distributed memory, the fastest currently known suffix array construction algorithm is based on prefix doubling~\cite{DBLP:conf/sc/FlickA15}.

\subparagraph*{Induced-Copying.}
Induced-copying algorithms sort a (small) subset of suffixes and then \emph{induce} the order of all other suffixes using the subset of sorted suffixes.
First, all suffixes are classified using one of two~\cite{ItohT1999ABCopySACA,NongZH2011SADS} classification schemes.
Here, all suffixes that have to be manually sorted are in a special class.
Then, the classification allows us to induce all non-special suffixes based on their class, starting characters, and preceeding or succeeding special-class suffix.
The inducing part of these algorithms usually consists of just two scans of the text, where for each position only one or two characters have to be compared.
Combined with a recursive approach, induced copying algorithms can compute the suffix array in linear time requiring only constant  working space in addition to the space for the suffix array~\cite{Goto2019OptSACA,LiLH18OptSACA}.
This combination is also very successful, as it is used by the fastest sequential suffix array construction algorithms~\cite{BahneBBBFFGKLMP2019SACABench,FischerK2017DismantlingDivSufSort,Grebnov2021libsais,Mori2015libdivsufsort}
Interestingly, there is only one linear time suffix array construction algorithm based on induced copying that does not also rely on recursion~\cite{Baier2016GSACA}.
In distributed memory, induced copying algorithms are space-efficient~\cite{DBLP:conf/alenex/0001K19}.

\subparagraph*{Recursive Algorithms.}
The third and final technique is to use recursion to solve subproblems of ever decreasing sizes.
Here, the general idea is to partition the input into multiple different (potentially overlapping) substrings.
A subset of these substrings can then be sorted using an integer sorting algorithm (in linear time).
If all substrings are unique, we can compute a suffix array together with the remaining suffixes not yet sorted.
Otherwise, we recurse on the non-unique ranks of the substrings as new input.
We can then use the suffix array from the recursive problem to compute the unique ranks from the original subset of substrings.
The first linear time suffix array construction algorithm is purely based on recursion~\cite{DBLP:journals/jacm/KarkkainenSB06}.
This algorithm is also the foundation of the distributed-memory suffix array construction algorithm presented in this paper.
It already has been considered in distributed memory~\cite{bingmann2018pdcx,DBLP:conf/bigdataconf/BingmannGK18}.
However, all implementations are straightforward transformations of the sequential algorithm to distributed memory.
We also want to mention that all but one~\cite{Baier2016GSACA} suffix array construction algorithm with linear running time at least partially utilizes this recursive principle of solving a smaller subproblem using the same algorithm.

\section{A Space-Efficient Variant of Distributed DCX}
\label{sec:distributed-dcx}
In this section, describe the general idea of the sequential DC3 algorithm~\cite{DBLP:journals/jacm/KarkkainenSB06}.
Then, we describe a canonical transformation of the sequential DC3 algorithm to a distributed-memory algorithm.
Here, we also consider the more general form---the DCX algorithm.
Finally, we discuss how to optimize this canonical transformation to a scaling, fast, and memory-efficient distributed suffix array construction algorithm.

\subsection{The Sequential DCX Algorithm}
\label{sec:sequential-dcx}
The \emph{skew} or \textbf{D}ifference \textbf{C}over3 algorithm (DC3) -- and its generalization \DCX{} -- is a recursive suffix array construction algorithm which exhibits linear running time (in the sequential setting).
As we propose a fast and lightweight distributed variant of this algorithm as our main contribution, we briefly discuss its main ideas.
The \DCX{} algorithm uses so-called difference covers to partition the suffixes of the input text into $X$ different sets.
A difference cover $D_{X}$ modulo $X$ is a subset of $[0, X)$ such that for all $i,j \in \mathbb{N}$, there is a $0 \le l < X$ with $(i + l) \modop X \in D_X$ and $(j + l) \modop X \in D_X$.
Put differently, $[0, X) = \{i-j \modop X \mid i,j \in D_X\}$.
$X = 3$ is the smallest $X$ for which a non-trivial difference cover exists with $D_3 = \{1,2\}$.

Suffixes with index $(j \modop X) \in D_X$ constitute the \emph{(difference-cover) sample}.
For now, let us assume that we already know a relative ordering of the sample suffixes within the final suffix array.
For any two suffixes $s_i$ and $s_j$, there is an $l < X$ such that $(i+l)$ and $(j+l)$ are indices of sample suffixes.
Hence, for lexicographically comparing $s_i$ and $s_j$ it is sufficient to compare the pairs $(T[i,i+l), \mathrm{rank}[i+l])$ and $(T[j,j+l), \mathrm{rank}[j+l])$.
For $X=3$, this rank-\emph{inducing} can be achieved using linear-time integer sorting and merging.
A relative ordering of the samples can be computed by sorting their $X$-prefixes, replacing them with the rank of their prefix, and recursively applying the algorithm to this auxiliary text (see \Cref{sec:distributed-dcx} for more details).
For DC3, the number of sample suffixes is $\le 2/3n$, and as all other operations can be achieved with work linear in the size of the input, the overall complexity of the DC3 algorithm is also in $\mathcal{O}(n)$.

It remains to discuss how a relative ordering of the sample suffixes is determined.
If the $X$-prefixes of the sample suffixes are unique, we are already done, as we can take their rank for the ordering.
Otherwise, we replace the sample suffixes $j$ with the rank of their $X$-prefix and order them by $(j \modop X, j \divop X)$.
Recursively apply the algorithm to this text $T'$ yields a suffix array $SA'$ from which we can retrieve a relative ordering of the sample suffix with regard to the original text $T$.

\subsection{The Distributed DCX Algorithm}
\label{sec:distributed-dcx}

Our distributed suffix array construction is a simple and practical distributed variant of the \DCX{} algorithm for $X \ge 3$.
\Cref{algo:ddcx} shows a high-level pseudocode for the algorithm.

\newlength{\lencomment}
\setlength{\lencomment}{2.1in}
\begin{algorithm}[t]
\DontPrintSemicolon
\SetAlgoNlRelativeSize{0}
\SetKwFunction{RadixSort}{RadixSort}
\SetKwFunction{PrefixSum}{PrefixSum}
\SetKwFunction{PrefixSum}{PrefixSum}
\SetKwFunction{RecursiveDCThree}{RecursiveDC3}
\SetKwFunction{ComputeSuffixArray}{ComputeSuffixArray}
\SetKwFunction{Sort}{Sort}
\SetKwFunction{RetrieveRank}{RetrieveRank}
\SetKwFunction{Rank}{rank}
\SetKwFunction{Origin}{origin}

\KwIn{Text $T_i$ on PE $i$.}
\KwOut{Local Chunk of distributed Suffix Array of $T$.}

    $o_i = \PrefixSum(\abs{T_i})$    \tcp*{ \parbox[t]{\lencomment}{global text index offset}}
    $C_i = \langle 0 \le j < \abs{T_i} \mid (j + o_i \modop X) \in D_X\rangle$ \tcp*{\parbox[t]{\lencomment}{difference cover sample}}
    $S_i = \langle (\underbrace{T_i[j,j + X)}_{X\textnormal{-prefix}}, \underbrace{j + o_i}_{\mathrm{global\ idx}}) \mid j + o_i \in C_i \rangle$ \tcp*{\parbox[t]{\lencomment}{\raggedright (prefix,idx)-pair of dif\-fer\-ence cover sample suffixes}}
    globally sort $S_i$ by first entry\;
    
    \If{all first entries of $S$ are unique} {
	\For{$t = (\textrm{prefix},j) \in S_i$} {
	send $(j, \Rank(t, S))$ to PE $\Origin{j}$
	}
	store received rank data in $R_i$\;%
    }\Else{
	$P_i = \langle (R_i[j], j) \mid   j \in C_i \rangle$ \tcp*{\parbox[t]{\lencomment}{replace $X$-prefix of sample suffixes with rank}}
	globally sort $P_i$ by $(j \modop X, j \divop X)$\;
	recursively call DCX on $T_i' =\langle r \mid (r,j) \in P_i \rangle$\;
	\For{$j \in C_i$} {
	    retrieve (unique) rank of $j$ from suffix array of $T'$ and store in $R_i$\;
    	}
    }
    \For{$0 \le k < X$} {
	construct $S_i^k = \langle(T_i[j, \dots, j+X), R_i[j+k_1], \dots, R_i[j+k_v], j + o_i)) \mid 0 \le j < \abs{T_i}\rangle$\;
    }
    
    globally sort $S_i = S_i^0 \cup \dots \cup S_i^k$ by appropriate comparison function (see \cite{DBLP:journals/jacm/KarkkainenSB06})\;
    output last entry of $S_i$ as suffix array $SA_i$

\caption{High-level overview of a simple distributed variant of the DCX algorithm.}
    \label{algo:ddcx}
\end{algorithm}

We now discuss the algorithm in some more detail.
The input to the algorithm on PE $i$ is the local chunk $T_i$ of the input text $T$.

\begin{enumerate}
    \item \textbf{Sort the Difference Cover Sample}
	In the first phase of the algorithm, we select on each PE $i$ the suffixes starting at (global) positions $j$ with $(j \modop X) \in D_X$.
	These suffixes constitute the so-called \emph{difference cover sample}.
	The main idea of the algorithm is to compute the ranks of these suffixes first.
        For that we globally sort the $X$-prefixes of all suffixes within the difference cover sample.
        If all of them are unique, this already constitutes the relative ordering of the sample suffixes within the final suffix array.
        This rank information can then be used to rank any two suffixes $s_i$ and $s_j$ (see \Cref{sec:sequential-dcx} and the following step three for details in the distributed setting).
        Otherwise we have to recurse on the sample suffixes as described in the following step of the algorithm.

    \item \textbf{Compute Unique Ranks Recursively:}
	If the ranks are not already unique, we locally create an array $P_i$ by replacing each entry $(X-\mathrm{prefix}, j)$ of the sorted sample suffix array $S_i$ with $(\mathrm{rank}[j], j)$, i.e., we replace each sample suffix with its previously computed rank.
	Afterwards, we globally sort $P_i$ by $(j\modop X, j\divop X)$.
	This rearranges the newly renamed sample suffixes in their original order by respecting the equivalence class of their starting index within $D_X$.
	We then recursively call the DCX algorithm on the text $T_i'$ where $T_i'$ simply contains the new names of the sample suffixes from $P_i$ dropping the index.
	From the suffix array of $T_i'$, we can easily determine the rank of each sample suffix $j$.
	Due to the construction of $T'$, the ranks of the sample suffixes correspond to their relative order in $T$.

    \item \textbf{Sort All Suffixes:}
	Now, we construct for each $0 \le k < X$ a set $S_i^k$ containing the $X$-prefixes of the suffixes $(j \modop X) = k$, $\abs{D_X}$ ranks and the index $j$.
	Sorting these sets $S_i^k$ altogether using the previously discussed comparison function for suffixes $s_i$, $s_j$ yields the suffix array of the original text $T$.

	Note that in the original work in the sequential setting, the sets $S_i^k$ are not sorted altogether but individually and later merged to ensure work in $\mathcal{O}(\abs{D_X}n)$.
\end{enumerate}

Existing implementations of distributed DC3(/DC7/DC13) \cite{DBLP:journals/pc/KullaS07, bingmann2018pdcx} broadly follow this scheme which is a straightforward adaption of the sequential algorithm to the distributed setting.
However, this approach is not particularly space-efficient.
Materializing the $X$-prefixes of the (non-)sample suffixes and sorting (or merging) them results in a memory blow-up proportional to $X$ compared to the actual input.
Consequently, sorting suffixes on real distributed machines using DCX with large $X$ does not seem feasible due to the limited main memory, even though DCX with $X > 3$ has a better performance on many real-world inputs \cite{DBLP:conf/alenex/0001K19}.

In the following \Cref{sec:bucketing}, we propose a technique to overcome this problem.

\subsubsection{Bucketing}
\label{sec:bucketing}
In the sequential or shared-memory parallel setting, $X$-prefixes of suffixes can be sorted space-efficiently as each such element $e$ can be represented as a pointer
to the starting position of the suffix within the input text.
This space-efficient sorting, however, is no longer possible in distributed memory.
If we want to globally sort a distributed array of suffix-prefixes, we have to fully materialize and exchange them -- resulting in a memory blow-up of at least a factor $X$.
A simple idea to prevent this blow-up is to use a partitioning strategy which divides the elements from the distributed array into multiple buckets using splitter elements and processes only one bucket at a time. 

In the following, we describe a general technique for space-efficient sorting which we proposed in our previous work on scalable distributed string sorting \cite{PascalStringSorting, DBLP:conf/esa/KurpiczM0S24}.
We use this generalized technique as a building block in our distributed variant of the \DCX{} algorithm.

Whenever a distributed array of elements with a space-efficient representation has to be globally sorted, we first determine $q$ global splitter elements $s_0, s_1, \dots s_{q-1}$ via sampling or multi-sequence selection with $s_0 = \infty$.
We then locally partition the array into $q$ buckets, such that element $e$ with $s_{k} < e \le s_{k+1}$ is placed in bucket $k$.
We then execute $q$ global sorting steps.
In each step $k$, we materialize and communicate the elements from bucket $k$ using a common distributed sorting algorithm.
Assuming that the splitters are chosen such that the global number of elements in each bucket is $n/q$ and the elements within each bucket are equally distributed among the PEs (see \Cref{sec:random-chunking} how this can be ensured), we only have to materialize $n/(pq)$ elements per bucket and PE instead of $n/p$ elements per PE when using only one sorting phase.

By choosing $q$ proportional to the memory blow-up caused by materializing an element, we can keep the overall space consumption of this distributed space-efficient sorting approach in $\mathcal{O}(n/p)$.

\subsubsection{Space-Efficient Randomization via Random Chunk Redistribution}
\label{sec:random-chunking}
The global number of elements per bucket can be balanced by judiciously choosing the splitter elements.
Using multi-sequence selection \cite{DBLP:conf/spaa/AxtmannBS015} one can obtain splitter elements balancing the global number of elements per bucket perfectly (up to rounding issues).
However, the number of elements per PE within a bucket can vary greatly depending on the input.
Assume an input which is already globally sorted with $q < p$ buckets.
In this setting, all $n/p$ elements located on the first PE have to be materialized when processing the first bucket.
This results in memory blow-up and poor load-balancing across the PEs.
Increasing the number of buckets $q$ can only address the memory consumption issue but does not help with load-balancing at all.

A standard technique to resolve this kind of problem is a random redistribution of the elements to be sorted.
However, this is not directly possible for elements which are stored in a space-efficient manner as in our case.

We propose to solve this problem by randomly redistributing not single prefixes of suffixes but whole chunks of the input text (together with some book-keeping information) before running the actual algorithm.

\begin{theorem}[Random Chunk Redistribution]
\label{thm:chunk-redistribution}
    When redistributing chunks of size $c$ uniformly at random across $p$ PEs, with $q$ buckets each containing $n/q$ elements, the expected number of elements from a single bucket received by a PE is $n/(pq)$.

    Furthermore, the probability that any PE receives $2n/(pq)$ or more elements from the same bucket is at most $1/p^\gamma$ for $n \ge 8c(\gamma+2) pq \ln(p)/3$ and $\gamma > 0$.

\end{theorem}

\begin{proof}
    Let $Y_i^k$ denote the number of elements belonging to bucket $k$ which are assigned to PE $i$.
    In the following, we will determine the expected value of $Y_i^k$ and show that $\mathbb{P}[Y_i^k \ge 2\mathbb{E}[Y_i^k]]$ is small.
    This will then be used to derive the above-stated bounds.

    Let $c_j^k$ be the number of elements belonging to bucket $k$ in chunk $j$.
    For the sake of simplicity, we assume all buckets to be of equal size, thus, $\sum_{j=0}^{n/c-1} c_j^k = n / q$.
    We define
    \[
	X_{j,i}^{k} =
        \begin{cases} 
	    c_j^k & \text{if chunk } j \text{ is assigned to PE } i\\
            0 & \text{otherwise,}
        \end{cases}
    \]
    for chunk $j$ with $0 \le j < n/c$, PE $i$ with $0 \le i < p$, and bucket $k$ with $0 \le k < q$.
    Thus, the random variable $X_{j,i}^k$ indicates the number of elements from bucket $k$ received by PE $i$ if chunk $j$ is assigned to this PE.
    Hence, we can express $Y_i^k$ as the sum over all $X_{j,i}^k$,i.e., $Y_i^k = \sum_{j=0}^{n/c -1} X_{j,i}^k$.
    As all chunks are assigned uniformly at random and there are $p$ PEs, we furthermore have $\mathbb{E}[X_{j,i}^{k}] = c_j^k/p$.
    By the linearity of expectation, we can derive the expected value of $Y_{i}^{k}$ as

    \begin{align*}
	\mathbb{E}[Y_{i}^{k}] 
	= \mathbb{E}\left[\sum_{j=0}^{n/c -1} X_{j,i}^k\right]
	= \sum_{j=0}^{n/c -1} \mathbb{E}[X_{j,i}^k] 
	= \sum_{j=0}^{n/c -1} \frac{c_{j}^k}{p} 
	= \frac{n}{pq}.
	\label{eq:prepare-bernstein}
    \end{align*}

    For each bucket $k$, we now bound the probability $\mathbb{P}[Y_{i}^{k} \ge 2n/(pq)]$ that PE $i$ receives two times its expected number of elements or more.
    We have
    \begin{align}
	\mathbb{P}\left[Y_{i}^{k} \ge \frac{2n}{pq}\right] &= \mathbb{P}\left[\sum_{j=0}^{n/c -1} X_{j,i}^k \ge \frac{2n}{pq}\right] \notag \\
	&= \mathbb{P}\left[\sum_{j=0}^{n/c -1} X_{j,i}^k - \frac{n}{pq} \ge \frac{n}{pq}\right]
	= \mathbb{P}\left[\sum_{j=0}^{n/c -1} X_{j,i}^k - \mathbb{E}[X_{j,i}^k] \ge \frac{n}{pq}\right].
    \end{align}

    As the value of $X_{i,j}^k$ is bounded by the chunk size $c$, the Bernstein~inequality
    ~
    \cite[Theorem 2.10, Corollary 2.11]{DBLP:books/daglib/0035704} yields the following bound
    \begin{equation}
	\mathbb{P}\left[\sum_{j=0}^{n/c -1} X_{j,i}^k - \mathbb{E}[X_{j,i}^k] \ge \frac{n}{pq}\right] \le \exp\left(-\frac{\left(\frac{n}{pq}\right)^2}{2\left(\sum_{j=0}^{n/c-1} \mathbb{E}[(X_{j,i}^k)^2] + \frac{cn}{3pq}\right)}\right).
	\label{eq:bernstein}
    \end{equation} 
    Since we find $E[(X_{j,i}^k)^2] = (c_j^k)^2/p$, it follows that
    \[
      \sum_{j=0}^{n/c-1} \mathbb{E}[(X_{j,i}^k)^2] = \sum_{j=0}^{n/c-1} (c_j^k)^2/p \le \frac{1}{p} \sum_{j=0}^{n/(qc)-1} c^2 = \frac{cn}{pq},
    \]
    as the sum of the squares of a set of elements $0 \le a_i \le c$ with $\sum_i a_i = b$ and $b$ divisible by $c$  is maximized if they are distributed as unevenly as possible, i.e., $a_i = c$ for $b/c$ elements and $0$ for all others.
    We can use this estimation for an upper bound on the right-hand side of \eqref{eq:bernstein}
    \begin{equation}
	\exp\left(-\frac{\left(\frac{n}{pq}\right)^2}{2\left(\sum_{j=0}^{n/c-1} \mathbb{E}[(X_{j,i}^k)^2] + \frac{cn}{3pq}\right)}\right) \le
	\exp\left(-\frac{\left(\frac{n}{pq}\right)^2}{2\left(\frac{cn}{pq} + \frac{cn}{3pq}\right)}\right) =
	\exp\left(-\frac{3n}{8pqc}\right).
	\label{eq:bernstein-finalize}
    \end{equation}
    Combining these estimations, we obtain the bound
    \begin{equation*}
	\mathbb{P}\left[Y_{i}^{k} \ge \frac{2n}{pq}\right] \overset{\eqref{eq:prepare-bernstein}, \eqref{eq:bernstein-finalize}}{\le} \exp\left(-\frac{3n}{8pqc}\right) \le \exp\left(-(\gamma+2)\ln{p}\right) = \frac{1}{p^{\gamma + 2}}
    \end{equation*}
    for $n \ge 8pqc\ln(p)(\gamma + 2)/3$.

    Although the random variables $Y_i^k$ are dependent on each other, using the union-bound argument yields the following estimation
    \[
	\doublestroke{P}\left[ \bigcup Y_i^k \ge 2\frac{n}{pq} \right] \le \sum_{i=0}^{p-1} \sum_{k=0}^{q-1} \doublestroke{P}[Y_i^k \ge 2\frac{n}{pq}] \le \sum_{i=0}^{p-1} \sum_{k=0}^{q-1} \frac{1}{p^{\gamma+2}} \le \frac{1}{p^{\gamma}}.
    \]
    Hence,  we obtain $\frac{1}{p^\gamma}$ as an upper bound on the probability that any PE receives more than two times the expected number of elements $n/(pq)$  for any bucket when assuming $q \le p$.
\end{proof}

\Cref{thm:chunk-redistribution} shows that combining a random chunk redistribution with our bucketing approach yields a space-efficient solution to the sorting problems occurring within our distributed variant of the \DCX{} algorithm with provable performance guarantees.

Note that in the \DCX{} algorithm one can either perform a single redistribution at the beginning of each level or apply a random chunk redistribution before executing a space-efficient sorting via bucketing step.
Depending on the actual implementation, one does not only send chunks of the text but also corresponding rank entries and additional book-keeping information like the global index of a chunk.
Furthermore, each chunk should contain an \emph{overlap} of $X$ characters to ensure that an $X$-prefix for each element within a chunk can be constructed without communication.

\subsubsection{Further Optimizations}
\label{sec:optimizations}
In addition to the techniques described above, we also utilize discarding and packing, two techniques commonly used in distributed and external memory suffix array construction algorithms.

\subparagraph*{Discarding.}
    After sorting the $X$-prefixes of the sample suffixes, we have to recursively apply the DCX algorithm (or any other suffix sorting algorithm) to a smaller subproblem if there are duplicate ranks.
    However, in order to obtain overall unique ranks for the sample suffixes, we do not have to recurse on all of them  but can discard suffixes whose ranks are already unique after initial sorting.
    This so-called \emph{discarding} technique has been proposed for and has been implemented in the external memory setting~\cite{DBLP:journals/jea/DementievKMS08} but to the best of our knowledge has not been explored for distributed memory yet.
\subparagraph*{Packing}
    Packing is an optimization for small-sized alphabets proposed for distributed memory prefix-doubling by Flick et al.~\cite{DBLP:conf/sc/FlickA15}.
    Assume $b = \lceil\log{\sigma}\rceil < B$, where $B$ is the size of one machine word.
    Instead of using one machine word per character of the $X$-prefix, we can instead consider packing $\lfloor\frac{B X}{b}\rfloor$ characters into $X$ machine words  or use $X$ characters in only $\lceil\frac{Xb}{B}\rceil$ words.

\section{Extension to the Distributed External Memory Model}
\label{sec:external-memory}
    Our bucketing technique (together with the randomized chunking approach) can be adapted to the distributed external-memory model, where each PE has a main memory of size $M$ and additional external-memory (disk) storage from which blocks of size $B$ words can be read at a time.
    In the following, we assume that the input text $T_i$ and the corresponding suffix array to be computed are located in external memory on each PE $i$.
    Whenever we want to sort elements stored in external memory during the algorithm using $q$ buckets, we first scan blockwise through the text (or associated information) to construct a set of sample elements.
    These samples are then globally sorted and $q-1$ splitters are equidistantly drawn and communicated to all PEs.
    The splitters are kept in main memory.
    Afterwards, for processing each bucket $k < q$, we again scan the input text blockwise from disk and keep the elements belonging to bucket $k$ in memory.
    Note that by judiciously choosing the number of sample elements for splitter construction, we can enforce the number of elements belonging to a bucket to be in $\mathcal{O}(pM)$ with high probability.

    To ensure that the number of elements on each PE belonging to a bucket is in $\mathcal{O}(M)$, we can apply our randomized chunking technique.
    For this, we read $\mathcal{O}(M)$-sized parts of the input text into main memory at a time, apply the in-memory random chunking-based redistribution as described in \Cref{sec:random-chunking}, and write the received chunks to disk.

\section{Preliminary Implementation and Evaluation}
\label{sec:evaluation}

For a first preliminary evaluation, we use up to $128$ compute nodes of SuperMUC-NG, where each node is equipped with an Intel Skylake Xeon Platinum 8174 processor with 48 cores and 96GB of main memory.
The internal interconnect is a fast OmniPath network with 100 Gbit/s.

We use inputs from three different data sets:
\begin{itemize}
  \item \textbf{CommonCrawl (\CC{}).} This input consists of websites crawled by the Common Crawl Project.
We use the \emph{WET} files, which contain only the textual data of the crawled websites, i.\,e., no HTML tags.
Furthermore, we removed the meta information added by the Commoncrawl corpus.
We used the following WET files: \url{crawl-data/CC-MAIN-2019-09/segments/1550247479101.30/wet/CC-MAIN-20190215183319-20190215205319-#ID.warc.wet}, where \texttt{\#ID} is in the range from $00000$ to $000600$.
  \item \textbf{Wikipedia (\Wiki{}).} This file contains the XML data of all pages in the most current version of the Wikipedia, i.\,e., the files available at \url{https://dumps.wikimedia.org/#IDwiki/20190320/#IDwiki-20190320-}\url{pages-meta-current.xml.bz2}, where \texttt{\#ID} is \emph{de}, \emph{en}, \emph{es}, and \emph{fr}.
  \item \textbf{DNA data (\DNA{}).} Here, we extract the DNA data from FASTQ files provided by the \emph{1000 Genomes Project}.
We discarded all lines but the DNA data and cleaned it, such that it only contains the characters \texttt{A}, \texttt{C}, \texttt{G}, and \texttt{T}. (We simply removed all other characters.)
The original FASTQ files are available at \url{ftp://ftp.sra.ebi.ac.uk/vol1/fastq/DRR000/DRR#ID}, where \texttt{\#ID} is in the range from $000001$ to $000426\_1$.
\end{itemize}

For this evaluation, we compare our (preliminary) distributed DCX implementation (using X=21, with packing and discarding optimizations) with the current state-of-the-art distributed suffix array construction algorithm \PSAC{}~\cite{DBLP:conf/sc/FlickA15}.
Both algorithms are implemented in C++ and use MPI for interprocess communication.
Additionally, our implementation uses the (zero-overhead) MPI Wrapper KaMPIng~\cite{kamping-sc}.

\begin{figure}[t]
    \begin{center}
        \hspace*{-2.5cm}
        \scalebox{0.6}{\input{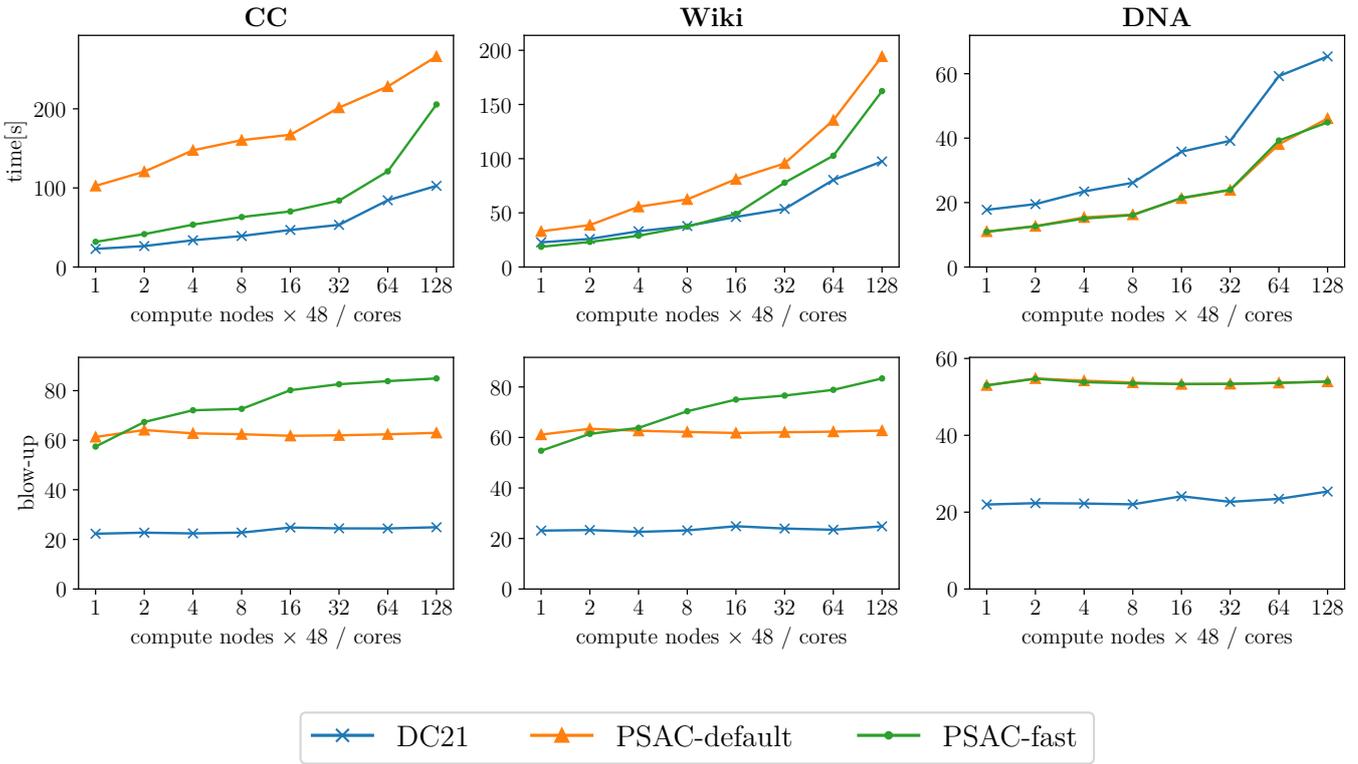}}
    \end{center}
    \begin{center}
        \scalebox{0.8}{
\begingroup%
\makeatletter%
\begin{pgfpicture}%
\pgfpathrectangle{\pgfpointorigin}{\pgfqpoint{5.200000in}{0.500000in}}%
\pgfusepath{use as bounding box, clip}%
\begin{pgfscope}%
\pgfsetbuttcap%
\pgfsetmiterjoin%
\definecolor{currentfill}{rgb}{1.000000,1.000000,1.000000}%
\pgfsetfillcolor{currentfill}%
\pgfsetlinewidth{0.000000pt}%
\definecolor{currentstroke}{rgb}{1.000000,1.000000,1.000000}%
\pgfsetstrokecolor{currentstroke}%
\pgfsetdash{}{0pt}%
\pgfpathmoveto{\pgfqpoint{0.000000in}{0.000000in}}%
\pgfpathlineto{\pgfqpoint{5.200000in}{0.000000in}}%
\pgfpathlineto{\pgfqpoint{5.200000in}{0.500000in}}%
\pgfpathlineto{\pgfqpoint{0.000000in}{0.500000in}}%
\pgfpathlineto{\pgfqpoint{0.000000in}{0.000000in}}%
\pgfpathclose%
\pgfusepath{fill}%
\end{pgfscope}%
\begin{pgfscope}%
\pgfsetbuttcap%
\pgfsetmiterjoin%
\definecolor{currentfill}{rgb}{1.000000,1.000000,1.000000}%
\pgfsetfillcolor{currentfill}%
\pgfsetfillopacity{0.800000}%
\pgfsetlinewidth{1.003750pt}%
\definecolor{currentstroke}{rgb}{0.800000,0.800000,0.800000}%
\pgfsetstrokecolor{currentstroke}%
\pgfsetstrokeopacity{0.800000}%
\pgfsetdash{}{0pt}%
\pgfpathmoveto{\pgfqpoint{0.069992in}{0.097222in}}%
\pgfpathlineto{\pgfqpoint{5.130008in}{0.097222in}}%
\pgfpathquadraticcurveto{\pgfqpoint{5.168897in}{0.097222in}}{\pgfqpoint{5.168897in}{0.136111in}}%
\pgfpathlineto{\pgfqpoint{5.168897in}{0.391666in}}%
\pgfpathquadraticcurveto{\pgfqpoint{5.168897in}{0.430555in}}{\pgfqpoint{5.130008in}{0.430555in}}%
\pgfpathlineto{\pgfqpoint{0.069992in}{0.430555in}}%
\pgfpathquadraticcurveto{\pgfqpoint{0.031103in}{0.430555in}}{\pgfqpoint{0.031103in}{0.391666in}}%
\pgfpathlineto{\pgfqpoint{0.031103in}{0.136111in}}%
\pgfpathquadraticcurveto{\pgfqpoint{0.031103in}{0.097222in}}{\pgfqpoint{0.069992in}{0.097222in}}%
\pgfpathlineto{\pgfqpoint{0.069992in}{0.097222in}}%
\pgfpathclose%
\pgfusepath{stroke,fill}%
\end{pgfscope}%
\begin{pgfscope}%
\pgfsetrectcap%
\pgfsetroundjoin%
\pgfsetlinewidth{1.505625pt}%
\definecolor{currentstroke}{rgb}{0.121569,0.466667,0.705882}%
\pgfsetstrokecolor{currentstroke}%
\pgfsetdash{}{0pt}%
\pgfpathmoveto{\pgfqpoint{0.108881in}{0.281944in}}%
\pgfpathlineto{\pgfqpoint{0.303326in}{0.281944in}}%
\pgfpathlineto{\pgfqpoint{0.497770in}{0.281944in}}%
\pgfusepath{stroke}%
\end{pgfscope}%
\begin{pgfscope}%
\pgfsetbuttcap%
\pgfsetroundjoin%
\definecolor{currentfill}{rgb}{0.121569,0.466667,0.705882}%
\pgfsetfillcolor{currentfill}%
\pgfsetlinewidth{1.003750pt}%
\definecolor{currentstroke}{rgb}{0.121569,0.466667,0.705882}%
\pgfsetstrokecolor{currentstroke}%
\pgfsetdash{}{0pt}%
\pgfsys@defobject{currentmarker}{\pgfqpoint{-0.041667in}{-0.041667in}}{\pgfqpoint{0.041667in}{0.041667in}}{%
\pgfpathmoveto{\pgfqpoint{-0.041667in}{-0.041667in}}%
\pgfpathlineto{\pgfqpoint{0.041667in}{0.041667in}}%
\pgfpathmoveto{\pgfqpoint{-0.041667in}{0.041667in}}%
\pgfpathlineto{\pgfqpoint{0.041667in}{-0.041667in}}%
\pgfusepath{stroke,fill}%
}%
\begin{pgfscope}%
\pgfsys@transformshift{0.303326in}{0.281944in}%
\pgfsys@useobject{currentmarker}{}%
\end{pgfscope}%
\end{pgfscope}%
\begin{pgfscope}%
\definecolor{textcolor}{rgb}{0.000000,0.000000,0.000000}%
\pgfsetstrokecolor{textcolor}%
\pgfsetfillcolor{textcolor}%
\pgftext[x=0.653326in,y=0.213889in,left,base]{\color{textcolor}\rmfamily\fontsize{14.000000}{16.800000}\selectfont DC21}%
\end{pgfscope}%
\begin{pgfscope}%
\pgfsetrectcap%
\pgfsetroundjoin%
\pgfsetlinewidth{1.505625pt}%
\definecolor{currentstroke}{rgb}{1.000000,0.498039,0.054902}%
\pgfsetstrokecolor{currentstroke}%
\pgfsetdash{}{0pt}%
\pgfpathmoveto{\pgfqpoint{1.529037in}{0.281944in}}%
\pgfpathlineto{\pgfqpoint{1.723482in}{0.281944in}}%
\pgfpathlineto{\pgfqpoint{1.917926in}{0.281944in}}%
\pgfusepath{stroke}%
\end{pgfscope}%
\begin{pgfscope}%
\pgfsetbuttcap%
\pgfsetmiterjoin%
\definecolor{currentfill}{rgb}{1.000000,0.498039,0.054902}%
\pgfsetfillcolor{currentfill}%
\pgfsetlinewidth{1.003750pt}%
\definecolor{currentstroke}{rgb}{1.000000,0.498039,0.054902}%
\pgfsetstrokecolor{currentstroke}%
\pgfsetdash{}{0pt}%
\pgfsys@defobject{currentmarker}{\pgfqpoint{-0.041667in}{-0.041667in}}{\pgfqpoint{0.041667in}{0.041667in}}{%
\pgfpathmoveto{\pgfqpoint{0.000000in}{0.041667in}}%
\pgfpathlineto{\pgfqpoint{-0.041667in}{-0.041667in}}%
\pgfpathlineto{\pgfqpoint{0.041667in}{-0.041667in}}%
\pgfpathlineto{\pgfqpoint{0.000000in}{0.041667in}}%
\pgfpathclose%
\pgfusepath{stroke,fill}%
}%
\begin{pgfscope}%
\pgfsys@transformshift{1.723482in}{0.281944in}%
\pgfsys@useobject{currentmarker}{}%
\end{pgfscope}%
\end{pgfscope}%
\begin{pgfscope}%
\definecolor{textcolor}{rgb}{0.000000,0.000000,0.000000}%
\pgfsetstrokecolor{textcolor}%
\pgfsetfillcolor{textcolor}%
\pgftext[x=2.073482in,y=0.213889in,left,base]{\color{textcolor}\rmfamily\fontsize{14.000000}{16.800000}\selectfont PSAC-default}%
\end{pgfscope}%
\begin{pgfscope}%
\pgfsetrectcap%
\pgfsetroundjoin%
\pgfsetlinewidth{1.505625pt}%
\definecolor{currentstroke}{rgb}{0.172549,0.627451,0.172549}%
\pgfsetstrokecolor{currentstroke}%
\pgfsetdash{}{0pt}%
\pgfpathmoveto{\pgfqpoint{3.645412in}{0.281944in}}%
\pgfpathlineto{\pgfqpoint{3.839856in}{0.281944in}}%
\pgfpathlineto{\pgfqpoint{4.034301in}{0.281944in}}%
\pgfusepath{stroke}%
\end{pgfscope}%
\begin{pgfscope}%
\pgfsetbuttcap%
\pgfsetroundjoin%
\definecolor{currentfill}{rgb}{0.172549,0.627451,0.172549}%
\pgfsetfillcolor{currentfill}%
\pgfsetlinewidth{1.003750pt}%
\definecolor{currentstroke}{rgb}{0.172549,0.627451,0.172549}%
\pgfsetstrokecolor{currentstroke}%
\pgfsetdash{}{0pt}%
\pgfsys@defobject{currentmarker}{\pgfqpoint{-0.020833in}{-0.020833in}}{\pgfqpoint{0.020833in}{0.020833in}}{%
\pgfpathmoveto{\pgfqpoint{0.000000in}{-0.020833in}}%
\pgfpathcurveto{\pgfqpoint{0.005525in}{-0.020833in}}{\pgfqpoint{0.010825in}{-0.018638in}}{\pgfqpoint{0.014731in}{-0.014731in}}%
\pgfpathcurveto{\pgfqpoint{0.018638in}{-0.010825in}}{\pgfqpoint{0.020833in}{-0.005525in}}{\pgfqpoint{0.020833in}{0.000000in}}%
\pgfpathcurveto{\pgfqpoint{0.020833in}{0.005525in}}{\pgfqpoint{0.018638in}{0.010825in}}{\pgfqpoint{0.014731in}{0.014731in}}%
\pgfpathcurveto{\pgfqpoint{0.010825in}{0.018638in}}{\pgfqpoint{0.005525in}{0.020833in}}{\pgfqpoint{0.000000in}{0.020833in}}%
\pgfpathcurveto{\pgfqpoint{-0.005525in}{0.020833in}}{\pgfqpoint{-0.010825in}{0.018638in}}{\pgfqpoint{-0.014731in}{0.014731in}}%
\pgfpathcurveto{\pgfqpoint{-0.018638in}{0.010825in}}{\pgfqpoint{-0.020833in}{0.005525in}}{\pgfqpoint{-0.020833in}{0.000000in}}%
\pgfpathcurveto{\pgfqpoint{-0.020833in}{-0.005525in}}{\pgfqpoint{-0.018638in}{-0.010825in}}{\pgfqpoint{-0.014731in}{-0.014731in}}%
\pgfpathcurveto{\pgfqpoint{-0.010825in}{-0.018638in}}{\pgfqpoint{-0.005525in}{-0.020833in}}{\pgfqpoint{0.000000in}{-0.020833in}}%
\pgfpathlineto{\pgfqpoint{0.000000in}{-0.020833in}}%
\pgfpathclose%
\pgfusepath{stroke,fill}%
}%
\begin{pgfscope}%
\pgfsys@transformshift{3.839856in}{0.281944in}%
\pgfsys@useobject{currentmarker}{}%
\end{pgfscope}%
\end{pgfscope}%
\begin{pgfscope}%
\definecolor{textcolor}{rgb}{0.000000,0.000000,0.000000}%
\pgfsetstrokecolor{textcolor}%
\pgfsetfillcolor{textcolor}%
\pgftext[x=4.189856in,y=0.213889in,left,base]{\color{textcolor}\rmfamily\fontsize{14.000000}{16.800000}\selectfont PSAC-fast}%
\end{pgfscope}%
\end{pgfpicture}%
\makeatother%
\endgroup%
}
    \end{center}
    \caption{Running times and blow-up of the SACAs in our weak scaling experiments with 20MB per PE.}
    \label{fig:weak_scaling}
\end{figure}

\Cref{fig:weak_scaling} presents the running times and memory blow-up of weak scaling experiments with 20MB\footnote{Here, we are currently limited by the memory consumption of our competitor.}
of text data per PE (960 MB per compute note.
By blow-up, we refer to the maximum peak memory aggregated over each node divided by the total input size on a node.

We run \PSAC{} in two configurations. \PSACdefault{} is the standard (more memory-efficient) configuration proposed by the authors performing prefix-doubling without discarding initially and then switching to prefix-doubling with discarding in later iterations.
\PSACfast{} runs their prefix-doubling with discarding algorithm immediately.
Our variant outperforms both \PSAC{}-variants on \CC{} on all evaluated numbers of PEs, and is faster on \Wiki{} from 8 compute nodes on.
While \PSACfast{} is considerably faster than \PSACdefault{}, it also requires more memory.
However, on \DNA{}, both \PSAC{}-variants perform equally well and are faster than our \DCX{} implemenation.
Nevertheless, our \DCX{} implementation requires significantly less memory on all inputs.
Note, however, that we currently use 5-byte integers for rank information and indices in our implemenation.
Our competitor \PSAC{} uses 8-byte integers by default.
We were not able to easily replace them with $5$-byte integers, but we are planning to do so as part of future work to enable a more fair comparison.
In the future, we also plan to compare our algorithm with dedicated space-efficient distributed suffix-array construction algorithms~\cite{DBLP:conf/alenex/0001K19}.

In the following, we discuss some more details of our current implementation.

\subparagraph{Implementation Details.}
We use the IPS40 algorithm~\cite{axtmann2017} for local and AMS~\cite{DBLP:conf/spaa/AxtmannBS015, AxtmannSanders17} for distributed sorting.

Currently, we apply the bucketing technique for sorting the sample and non-sample suffixes in the third phase of the algorithm (see \cref{sec:distributed-dcx}),
with 32 buckets on the top level, and 8 buckets on the first recursion level  (for $X=21$). In subsequent recursion levels, the input is small enough such that sorting is not required.
Exploring general thresholds for the number of buckets depending on the input/machine configuration is part of our future work.
For larger $X$, it might be necessary to apply bucketing also for sorting the sample suffixes in the first phase.

To compare multiple characters of byte-alphabets (\CC{} and \Wiki{}), we pack 24 characters into 3 machine words ($64$ bits) for DC21.
Exploiting the small alphabet size of the \DNA{} dataset (3 bits per character), we pack 42 characters into 2 machine words, using less space and resulting in more unique sample ranks.
We are currently examining the best time/space tradeoffs for the packing heuristic.
As the alphabet size grows in the recursive calls of \DCX{}, packing is only used on the top level.

We are also experimenting with dedicated distributed string sorting algorithms, however, first preliminary experiments reveal, that AMS combined with packing tends to be slightly faster.
Furthermore, we are also exploring larger values for $X$.
Again, DC21/DC31 seem to have the best performance on the evaluated input instances.
However, this might be different for inputs with other characteristics.

The discarding optimization creates small overheads. Therefore, we use it only when there is sufficient reduction potential.

\section{Conclusion and Future Work}
In this work, we present initial algorithmic ideas on using a bucketing technique in conjunction with randomized chunking to develop a fast and space-efficient distributed suffix sorting algorithm.
Additionally, we provide first results of a preliminary implementation of our ideas.
We are currently working on improving our implementation, incorporating further optimizations, and extending our algorithm to the distributed external-memory model as outlined in \Cref{sec:external-memory}.
In addition, we also plan to look at distributed multi-GPU suffix sorting which could also benefit from our bucketing technique.
Furthermore, we want to explore the effects of low-latency \emph{multi}-level distributed (string) sorting.
This could be especially useful for small input sizes, where we want to compute the distributed suffix array with low latency, or when scaling our algorithms to (much) larger numbers of processors.

Our bucketing technique can also be employed to a generalization of distributed prefix doubling, where we do not double the investigated prefix length $h$ in each iteration
but increase it by a factor of $X$.
Therefore, we have to construct (and sort) a tuple containing $X$ ranks \cite{DBLP:journals/jea/DementievKMS08}.
However, in contrast to distributed \DCX{}, the information required for the construction of the tuples is not PE local.
Therefore, one has to query the rank entries twice per iteration --  once for determining into which bucket a suffix $j$ belongs and then for the actual bucketwise sorting.
Hence, it is not immediately clear whether this approach would yield a fast practical algorithm.
Orthogonally, one can reduce the memory consumption of distributed prefix doubling by using an \emph{in-place} alltoall exchange for exchanging rank information, which can be in-turn realized, e.g., by a bucketing approach. The same applies to the sorting step.

\label{sec:conclusion}

\bibliography{lipics-v2021-sample-article}
\end{document}